\documentclass[12pt,aps,ams,amsfonts,nofootinbib]{revtex4}

\usepackage[dvips]{graphics}
\usepackage{graphicx}
\usepackage{amssymb}
\usepackage{subfigure}
\usepackage{stackrel}
\usepackage{enumitem}
\usepackage{amsmath}
\usepackage{amsfonts}
\usepackage{bm}
\usepackage{color}
\usepackage{verbatim} 
\usepackage{tcolorbox}
\usepackage[normalem]{ulem}
\usepackage{hyperref}
\usepackage{comment}

\newlist{subquestion}{enumerate}{1}
\setlist[subquestion,1]{label=(\alph*)}

\newcommand{\ket}[1]{\ensuremath{\left|{#1}\right\rangle}}
\newcommand{\bra}[1]{\ensuremath{\left\langle{#1}\right |}}

\newcommand{\Tr}[1]{\textrm{Tr}\left[#1\right]}

\newtheorem{theorem}{Theorem}[section]

\newenvironment{proof}[1][Proof]{\begin{trivlist}
		\item[\hskip \labelsep {\bfseries #1}]}{\end{trivlist}}

\newcommand{\beq}{\begin{equation}}
\newcommand{\eeq}{\end{equation}}
\newcommand{\bse}{\begin{subequations}}
	\newcommand{\ese}{\end{subequations}}\newcommand{\bea}{\begin{eqnarray}}
\newcommand{\eea}{\end{eqnarray}}
\newcommand{\bit}{\begin{itemize}}
	\newcommand{\eit}{\end{itemize}}
\newcommand{\bpmatrix}{\begin{pmatrix}}
	\newcommand{\epmatrix}{\end{pmatrix}}

\newcommand{\be}{\begin{equation}}
\newcommand{\ee}{\end{equation}}
\newcommand{\ben}{\begin{eqnarray}}
\newcommand{\een}{\end{eqnarray}}

\begin{document}

\title[GHT and distinguishability notions]{Generalized Holevo theorem and distinguishability notions}

\author{D. G. Bussandri$^{1,2}$, P. W. Lamberti$^{1,2}$}
\affiliation{$^1$Facultad de Matem\'atica, Astronom\'{\i}a, F\'{\i}sica y Computaci\'on, Universidad Nacional de C\'ordoba, Av. Medina Allende s/n, Ciudad Universitaria, X5000HUA C\'ordoba, Argentina}
\affiliation{$^2$Consejo Nacional de Investigaciones Cient\'{i}ficas y T\'ecnicas de la Rep\'ublica Argentina, Av. Rivadavia 1917, C1033AAJ, CABA, Argentina}

\begin{abstract}
We present a generalization of the Holevo theorem by means of distances used in the definition of distinguishability of states, showing that each one leads to an alternative Holevo theorem. This result involves two quantities: the distance based Holevo quantity and the generalized accessible information. Additionally, we apply the new inequalities to qubits ensembles showing that for the Kolmogorov notion of distinguishability (for the case of an ensemble of two qubits) the generalized quantities are equal. On the other hand, by using a known example, we show that the Bhattacharyya notion captures not only the non-commutativity of the ensemble but also its purity. 
\end{abstract} 
\maketitle

\section{Introduction}\label{sec:Intro}

One of the fundamental results of quantum information theory is the so-called Holevo theorem, which establishes an upper bound for the accessible information (the amount of classical information which can be reliably encoded into a collection of quantum states) \cite{Holevo,Nielsen&Chuang,Jaeger,ZeroError}. 

A natural measure of ``quantumness'' for a quantum ensemble is the non-commutativity (measured by $\mathcal{N}_c$ \cite{Guo}) and it connects the accessible information $I_{\mathcal{S}_r}$ with the Holevo bound $\mathcal{X}_{\mathcal{S}_r}$. Specifically, if the elements of the ensemble commute then the non-commutativity is null and it holds $I_{\mathcal{S}_r}=\mathcal{X}_{\mathcal{S}_r}$, otherwise, the inequality is strict: $\mathcal{X}_{\mathcal{S}_r}>I_{\mathcal{S}_r}$ \cite{Tan}. Additionally, although the Holevo quantity can be achieved through measurements on a large number of copies, it is regularly not tight in the single-copy measurements case \cite{RuiHan}. Some authors have raised the question of whether the Holevo bound can be improved, by setting an alternative inequality  which has a better performance \cite{Zycowsky2010,Giovannetti}. To give a closed answer to the previous question exceeds the goal of this work; instead we focus on showing that the Holevo bound is a particular case of a more general ones. 

The proof of the Holevo bound uses the strong subaditivity of the quantum entropy, which can be established from the monotonicity of the quantum relative entropy. Further the Holevo quantity $\mathcal{X}_{S_r}$ can be expressed in terms of the generalized Jensen-Shannon divergence among the elements of the ensemble used in the communication protocol \cite{Majtey2005}, insinuating a close relation between the Holevo bound and the notion of distances between quantum states. Besides, in \cite{Tamir} a Holevo-type bound was obtained using the Hilbert-Schmidt distance measure and in \cite{Sharma2013,Wilde2014} was proposed and studied a \textit{generalized Holevo information} using a geometrical approach involving generalized divergences $\mathcal{D}(\cdot||\cdot)$ between quantum states. These facts motivated us to seek to generalize the Holevo theorem employing distance measures as mathematical objects with specific properties, obtaining new inequalities between the \textit{distance based Holevo quantity} (DBHQ) and the \textit{generalized accessible information} (GAI) [cf. Sec. \ref{sec:NewIneq}]. 

In quantum mechanics the states of a system are not observable. Therefore, the only way to distinguish two quantum states is through the measurements of physical quantities. Thus, a criterion to distinguish quantum states by using measurement procedures requires a \textit{measure of distinguishability} \cite{Fusch}. In the present work, we consider the main notions of distinguishability used in the field of quantum cryptography \cite{Fusch}, obtaining for each of them a version of the Holevo theorem [cf. Sec. \ref{sec:QuantDist2}] which involves new quantities (DBHQ and GAI) with different interpretations and behaviors.

In the particular case of qubits ensembles, we show that the DBHQ and GAI coincide for every ensemble of two elements for the Kolmogorov notion of distinguishability [cf. Sec. \ref{sec:QE}]. Besides, we have studied the behavior of the previous quantities for Bhattacharyya notion of distinguishability in a well known example [cf. Sec. \ref{sec:Example}] showing a richer behavior than $I_{\mathcal{S}_r}$ and $\mathcal{X}_{\mathcal{S}_r}$.

\section{Theoretical framework\label{sec:theory}}
This section is divided into two parts. In the first one we set the general communication framework while in the second we enunciate the properties of the different distances measures used. \par
We will always consider quantum systems associated with finite-dimensional Hilbert spaces. The states will be described by density matrices belonging to $\mathcal{B}_1^+(\mathcal{H})$ (namely, the set of bounded, positive-semidefinite operators with unit trace) where $\mathcal{H}$ denotes a Hilbert space.

\subsection{Communication scheme and related quantities} \label{sec:CommunSch}
Let us consider a quantum system $Q$ associated with the Hilbert space $\mathcal{H}_Q$. This system is shared by two entities, commonly referred to as \textit{Alice} and \textit{Bob}. In a communication scheme, the former has a classical information source $X=\{x_0,\dots,x_n\}$, $n\in\mathbb{N}$, and $p_i = \textrm{Prob}(X=x_i)$ the probability of occurrence of the value $x_i$. If at random it turns $X=x_i$ Alice prepares a quantum state $\rho_i$ belonging  to a fix ensemble $\{\rho_0,\dots,\rho_n\} \subset\mathcal{B}_1^+(\mathcal{H}_Q) $. The central aim is to communicate the result $X=x_i$ to the other part by mean of the states $\{\rho_i\}$. Accordingly, Bob makes a measurement over the system $Q$ described by the POVM (Positive-Operator Valued Measure) $\mathcal{M}=\{M_0,\dots,M_m\}$, $m\in\mathbb{N}$.\par 
The measurement results constitute a new random variable $Y^\mathcal{M}=\{y_0,\dots,y_m\}$ with probability of occurrence given by $q_j=\textrm{Tr}[M_j\rho]$ being $\rho=\sum_i p_i \rho_i$.\par
The conditional probability of obtaining the measurement result $Y^\mathcal{M}=y_j$ given that $X=x_i$ is $q_{j|i}=\textrm{Tr}[\sqrt{M_j}\rho_i\sqrt{M_j}]$; therefore, the joint probability of the variables $X$ and $Y^\mathcal{M}$ is given by $P(X=x_i , Y^\mathcal{M}=y_j)=P_{ij}=p_i q_{j|i}$. \par
Reasonably, the information that Bob gains about $X$ depends on the particular measurement $\mathcal{M}$ and it is quantified by the \textit{mutual information} $I(X,Y^\mathcal{M})$ \cite{Sasaki}:
\begin{eqnarray}\label{eq:I(x,y)}
I(X,Y^\mathcal{M})=H(X)+H(Y^\mathcal{M})-H(X,Y^\mathcal{M}),
\end{eqnarray}
being 
\begin{eqnarray*}
	H(X)=-\sum_{i=0}^n p_i\log p_i=H(p)\\
	H(Y^\mathcal{M})=-\sum_{j=0}^m q_j\log q_j=H(q) \\
	H(X,Y^\mathcal{M})=-\sum_{i,j=0}^{n,m} P_{ij}\log P_{ij}=H(P).
\end{eqnarray*}
The \textit{accessible information} $I_{\mathcal{S}_r}(X,Y)$ is defined as the maximum of $I(X,Y^\mathcal{M})$ over the possible measurements:
\begin{eqnarray}\label{eq:InfoAcc}
I_{\mathcal{S}_r}(X,Y)=\max_{\mathcal{M}} I(X,Y^\mathcal{M}).
\end{eqnarray}
The Holevo theorem \cite{Holevo,Nielsen&Chuang} establishes:
\begin{eqnarray*}
	I_{\mathcal{S}_r}(X,Y)\leq S(\rho)-\sum_i p_i S(\rho_i) = \mathcal{X}_{S_r},
\end{eqnarray*}
being $S(\rho)=-\textrm{Tr}[\rho \log \rho]$ the von Neumann entropy and $\mathcal{X}_{S_r}$ the \textit{Holevo bound} or \textit{Holevo information}.

\subsection{Distance Measures} \label{sec:QuantDist}

A central issue of this work is the concept of distance measure. Particularly, we are interested in the closeness of quantum states, namely, static measures \cite{Nielsen&Chuang}. There exist many distance measures used in different quantum information areas and frameworks. This variety generally is due to the arbitrariness in the distance measures definition \cite{Nielsen&Chuang,Fusch}. \par
A \textit{distance measure} $d(\cdot||\cdot)$ is a functional $d:\mathcal{B}^+_1(\mathcal{H}) \times \mathcal{B}^+_1(\mathcal{H}) \rightarrow \mathbb{R}_{\geq 0}$ with the following properties:
\begin{enumerate}[label=\alph*)]
	\item \label{d1} {\em Non-negativity}: For any $\rho,\sigma \in \mathcal{B}^+_1(\mathcal{H})$, $d(\rho||\sigma) \geq 0$ 
	\item \label{d2} {\em Identity of indiscernibles}: $d(\rho||\sigma) = 0$ if and only if $\rho=\sigma$  
\end{enumerate}
It is important to note that if the functional $d(\cdot||\cdot)$ satisfies also the symmetry property
\begin{enumerate}[label=\alph*),resume]
	\item \label{d3} {\em Symmetry}: For any $\rho,\sigma \in \mathcal{B}^+_1(\mathcal{H})$, $d(\rho||\sigma) = d(\sigma||\rho)$ 
\end{enumerate}
then sometimes it is called \textit{quantum distance}. \par
In this paper, we are interested in a set of distance measures with two additional properties which will allow us to establish the new inequalities [cf. Sec. \ref{sec:NewIneq}]. These requirements are:
\begin{enumerate}[label=\alph*),resume]
	\item \label{d4} $d(\cdot||\cdot)$ is non-increasing under the action of a completely positive trace-preserving map $\mathcal{E}(\cdot)$, 
	\begin{eqnarray*}
		d(\mathcal{E}(\rho)||\mathcal{E}(\sigma))\leq d(\rho||\sigma).
	\end{eqnarray*}
\end{enumerate}
It is important to note that if a distance measure fulfil the previous property then it satisfied the following requirement \cite{Wilde2014}: 
\begin{description}
	\item[Restricted Additivity:] \label{d5}
	$d(\rho_1 \otimes \sigma||\rho_2 \otimes \sigma) = d(\rho_1 || \rho_2)$ where $\rho_1,\rho_2 \in \mathcal{B}_1^+(\mathcal{H}_A)$, $\sigma \in \mathcal{B}_1^+(\mathcal{H}_B)$ and $\mathcal{H}=\mathcal{H}_A\otimes\mathcal{H}_B$.
\end{description}
As we shall see, if $d(\cdot||\cdot)$ fulfils the following requirement \cite{Bussandri19} then the new inequalities (see Sec. \ref{sec:NewIneq}) can be expressed in a clearer way.
\begin{enumerate}[label=\alph*),resume]
	\item  \label{d6} \textit{Additional Property:}
	\begin{eqnarray*}
		d( \ \sum_i p_i \left|i\right>\left<i\right| \otimes \rho_i \ || \ \sum_j p_j \left|j\right>\left<j\right| \otimes \rho ) = \sum_k p_k d(\rho_i || \rho),
	\end{eqnarray*}
	being $\{\left|i\right>\}$ an orthonormal basis of $\mathcal{H}_A$,  $\{p_i\}$ a probability distribution and $\{\rho_i\}$ elements of $\mathcal{B}_1^+(\mathcal{H}_B)$ such that $\sum_i  p_i \rho_i = \rho$. $\mathcal{H}_A$ and $\mathcal{H}_B$ are finite dimensional Hilbert spaces.
\end{enumerate}
Some examples of well-known distance measures fulfilling the previous conditions are the \textit{trace distance}, the \textit{squared Bures/Hellinger distance}, the \textit{quantum Jensen-Shannon divergence} (QJSD) and the \textit{relative entropy} \cite{Nielsen&Chuang,Majtey2005,Vedral2002,ZycowskyLibro}.

\section{New inequalities} \label{sec:NewIneq}

A well-known result about the mutual information $I$ between two random variables $X$ and $Y$ is the intrinsic connection with the (Shannon) relative entropy $H_r$. Specifically, in our communication context, we can rewrite \eqref{eq:I(x,y)} in the form:
\begin{eqnarray} \label{eq:distinMutalinfo}
I(X,Y^\mathcal{M})&=&H(p)+H(q)-H(P)= \nonumber \\
&=&\sum_{i,j=0}^{n,m} P_{ij} \log \frac{P_{ij}}{p_i q_j}=H_r(P||p\times q),
\end{eqnarray}
where $H_r(P||p\times q)$ is the (Shannon) relative entropy between the joint probability distribution $P=\{P_{ij}\}$ and his uncorrelated counterpart $p\times q=\{p_i q_i\}$. Also, $H_r$ is a measure of \textit{distance} between the former probabilities distributions \cite{Vedral2002}. Therefore, the equality \eqref{eq:distinMutalinfo} point out a connection between \textit{distinguishability} and the information shared by the random variables $X$ and $Y^\mathcal{M}$. The Holevo quantity can also be rewritten using the von Neumann relative entropy $S_r(\rho_0||\rho_1)=\textrm{Tr} \left[ \rho_0(\log_2 \rho_0 - \log_2 \rho_1)\right]$ \cite{Vedral2002}, resulting:
\begin{eqnarray*}
	\mathcal{X}_{S_r}=S(\rho)-\sum_i p_i S(\rho_i)=\sum_i p_i S_r(\rho_i||\rho),
\end{eqnarray*}
with $\rho=\sum_i p_i \rho_i$. Therefore, the Holevo theorem takes the following expression:
\begin{eqnarray*}
	H_r(P||p\times q) \leq \sum_i p_i S_r(\rho_i||\rho).
\end{eqnarray*}
In other words, the \textit{distinguishability} between the joint probability distribution $P$ and $p\times q$ is lower or equal than the mean \textit{distinguishability} between the states $\{\rho_i\}$, but here the distance measure used is solely the von Neumann relative entropy. We shall now generalize the previous inequality for general distance measures within a particular set. To do this, we will deal with two auxiliary Hilbert spaces  $\mathcal{H}_P$ and $\mathcal{H}_M$ with orthonormal basis $\{\ket{i_P}\}_{i=0}^n$ and $\{\ket{j_M}\}_{j=0}^m$, respectively. 
\begin{theorem}\label{th:theorem}
	Consider a \textit{distance measure} $d(\cdot||\cdot)$, verifying the properties \ref{d1}, \ref{d2}, \ref{d4} and \ref{d6} (see Sec. \ref{sec:QuantDist}). Then:
	\begin{eqnarray}\label{eq:result}
	D(P || p \times q) \leq \sum_i p_i d(\rho_i || \rho)\equiv\mathcal{X}_d,
	\end{eqnarray}
	being 
	\begin{eqnarray*}
		D(P || p \times q)\equiv d( \rho_{cc} ||\rho_{cc}^P \otimes \rho_{cc}^{M} ),
	\end{eqnarray*}
	where
	\begin{eqnarray*}
		\rho_{cc}  = \sum_{ij=0}^{nm} P_{ij}\ket{i_{P}} \bra{i_{P}} \otimes \ket{j_M} \bra{j_M}, \\
		\rho_{cc}^P \otimes \rho_{cc}^{M}  =  \sum_{kl=0}^{nm} p_{k}q_{l} \ket{k_{P}} \bra{k_{P}} \otimes \ket{l_M} \bra{l_M}.
	\end{eqnarray*}
	We will call to $\mathcal{X}_d$ the \textit{distance based Holevo quantity.}
\end{theorem}
\textit{Digression about notation:} For any reasonable distance measure $d( \rho_{cc} ||\rho_{cc}^P \otimes \rho_{cc}^{M} )$ constitutes a statistical distance (or divergence) between the probability distributions $P$ and $p \times q$ (on the grounds that the states $\rho_{cc}$ and $\rho_{cc}^P \otimes \rho_{cc}^{M}$ commute and are diagonal states in the orthonormal base $\{\ket{ij}_{PM}\}$ of $\mathcal{H}_P\otimes \mathcal{H}_M$). Therefore, we choose $D(P || p \times q)=d(\rho_{cc} || \rho_{cc}^P \otimes \rho_{cc}^{M})$ to emphasize that.  On the other hand, $D(P || p \times q)$ depends on the measurement $\mathcal{M}$, while $\mathcal{X}_d$ depends only on the ensemble $\{p_i,\rho_i\}$. Therefore, bearing in mind \eqref{eq:InfoAcc}, we choose:
\begin{eqnarray}\label{eq:InfoAccNews}
I_{d}(X,Y)\equiv \max_{\mathcal{M}} D(P||p\times q).
\end{eqnarray}
We will call to $I_{d}(X,Y)$ the \textit{generalized accessible information}. The interpretation of the previous quantity clearly hinges on the distance measure used.

\begin{proof}
	Consider now the communication scheme established in Sec. \ref{sec:CommunSch} and the Hilbert spaces $\mathcal{H}_P$ and $\mathcal{H}_M$ previously introduced.
	
	Let us connect each result of $X$ with an element of the orthonormal basis of $\mathcal{H}_P$, $\ket{i_P}$. On the other hand, such events are associated with a specific preparation, e.g. if $X=x_i$ then $\rho=\rho_i$ (the state of the shared system $Q$). These links are represented in the state:
	\begin{eqnarray}\label{eq:rhoC}
	\rho_c=\sum_{i=0}^n p_i \ket{i_P} \bra{i} \otimes \rho_i.
	\end{eqnarray}
	Furthermore, Bob makes a measurement represented by the POVM $\mathcal{M}$ and his results are stored on the subsystem $M$ under the following prescriptions: When no measurement was made, he chooses the state $\ket{0_M}$. On the contrary, if  Bob obtains the measurement result $Y^\mathcal{M}=y_j$ then he associates this with an element of the orthonormal basis of $\mathcal{H}_M$: $\ket{j_M}$. These two prescriptions have associated the states:: 
	\begin{eqnarray}\label{eq:rho0}
	\rho_0=\sum_{i=0}^n p_i \ket{i_P}\bra{i_P} \otimes \rho_i \otimes \ket{0_M}\bra{0_M},\\\label{eq:rhoM}
	\rho_\mathcal{M}=\sum_{i,j=0}^{n,m} p_i \ket{i_P}\bra{i_P} \otimes \sqrt{M_j}\rho_i\sqrt{M_j} \otimes \ket{j_M}\bra{j_M}.
	\end{eqnarray}
	The state $\rho_\mathcal{M}$ can be obtained from the application of a trace-preserving quantum operation $\mathcal{E}_{\mathcal{M}}$ i.e. $\mathcal{E}_{\mathcal{M}}(\rho_0)=\rho_\mathcal{M}$ \cite{Nielsen&Chuang}.
	
	The state that describes the information acquired by Bob about the variable $X$ is:
	\begin{eqnarray}\label{eq:rhoCC}
	\rho_{cc}=\sum_{i,j=0}^{n,m} P_{ij} \ket{i_P}\bra{i_P} \otimes \ket{j_M}\bra{j_M},
	\end{eqnarray}
	
	The uncorrelated counterparts of the states \eqref{eq:rhoC}, \eqref{eq:rho0}, \eqref{eq:rhoM} and \eqref{eq:rhoCC} respectively are:
	\begin{eqnarray}\label{eq:uncorrstate1}
	\rho_c^P \otimes \rho_c^{Q}=\sum_{i,j} p_ip_j \ket{i_P} \bra{i_P} \otimes \rho_j ,
	\end{eqnarray}
	\begin{eqnarray}\label{eq:uncorrstate2}
	\rho_0^P \otimes \rho_0^{QM}=\sum_{i,j} p_ip_j \ket{i_P} \bra{i_P} \otimes \rho_j \otimes \ket{0_M} \bra{0_M},
	\end{eqnarray}
	\begin{eqnarray}\label{eq:uncorrstate3}
	\rho_\mathcal{M}^P \otimes \rho_\mathcal{M}^{QM}&=&\mathcal{E}_{\mathcal{M}}(\rho_0^P \otimes \rho_0^{QM})=\\
	&=&\sum_{i,j,k} p_ip_k \ket{i_P} \bra{i_P} \otimes \sqrt{M_j}\rho_k\sqrt{M_j} \otimes \ket{j_M} \bra{j_M},
	\end{eqnarray}
	\begin{eqnarray}\label{eq:uncorrstate4}
	\rho_{cc}^P \otimes \rho_{cc}^{M}=\sum_{i,j=0}^{n,m} p_i q_j \ket{i_P}\bra{i_P} \otimes \ket{j_M}\bra{j_M}.
	\end{eqnarray}
	Let us consider now a distance measure $d(\cdot||\cdot)$ which fulfils the properties \ref{d1}, \ref{d2}, \ref{d4} and \ref{d6}. Then, due to the restricted additivity property (see Sec. \ref{d5}), it is direct to see that
	\begin{eqnarray*}
		d(\rho_c || \rho_c^P \otimes \rho_c^Q) =d(\rho_0 || \rho_0^P \otimes \rho_0^{QM})  ,
	\end{eqnarray*}
	By using \ref{d4}, we have
	\begin{eqnarray*}
		d(\rho_0 || \rho_0^P \otimes \rho_0^{QM})  \geq d(\rho_\mathcal{M} || \rho_\mathcal{M}^P \otimes \rho_\mathcal{M}^{QM}),
	\end{eqnarray*}
	and, taking into account that the trace operation over a subsystem is a trace preserving-quantum operation per se \cite{Nielsen&Chuang}, it follows:
	\begin{eqnarray*}
		d(\rho_\mathcal{M} || \rho_\mathcal{M}^P \otimes \rho_\mathcal{M}^{QM}) \geq d(\rho_{cc} || \rho_{cc}^P \otimes \rho_{cc}^{M}).
	\end{eqnarray*}
	Finally, combining \eqref{eq:uncorrstate1}-\eqref{eq:uncorrstate4} and using the additional property \ref{d6} we obtain:
	\begin{eqnarray*}
		d(\rho_c || \rho_c^P \otimes \rho_c^Q) = \sum_i p_i d(\rho_i || \rho) \geq d(\rho_{cc} || \rho_{cc}^P \otimes \rho_{cc}^{M}) = D(P || p \times q).
	\end{eqnarray*}
	$\blacksquare$
\end{proof}
It is important to remark that the Holevo theorem is recovered if we choose as distance measure the relative entropy $S_r(\cdot||\cdot)$.

On the other hand, it is easy to prove that if the states $\{\rho_i\}$ commutes between them, then for any reasonable distance measure the equality in \eqref{eq:result} is achieved. On the contrary, it depends on the distance measure considered if an equality in \eqref{eq:result} implies the commutation of the states $\{\rho_i\}$. This is a point that deserves to be remarked:  the implications of the inequality depend markedly on the distance used.
\section{Cryptographic distinguishability measures}\label{sec:QuantDist2}
As we have seen in Sec. \ref{sec:QuantDist}, a distance measure is a mathematical object with specific properties, but it is well known that not all of them can be considered as distinguishability measures. The reason is that the only physical way to distinguish two quantum states is through a measurement process. In quantum mechanics, the events are intrinsically stochastic and therefore the measurement results are associated with a probability distribution. In consequence, if one has a criterion for distinguishing two probability distributions then it is possible to obtain a \textit{measure of distinguishability} between two quantum states measuring and optimizing over the possible measurements \cite{Nielsen&Chuang,Fusch}.

The main purpose of this section is to use the new inequality \eqref{eq:result} and apply it to specific distinguishability notions.

Following \cite{Fusch}, we shall consider three different notions of distinguishability that are of interest to quantum cryptography: Kolmogorov distance (K), Probability of error (PE)  and Bhattacharyya coefficient (B). We will not include in this analysis the Shannon distinguishability mainly because of it does not have a closed expression in the quantum realm.\par
\vspace{0.5cm}
\noindent The Kolmogorov distance (K) between two probability distributions $p$ and $q$ is defined by
\begin{eqnarray}\label{eq:Kolmo}
K(p||q)=\frac{1}{2} \sum_i \left| p_i-q_i \right|.
\end{eqnarray}
On the other hand, the Kolmogorov distance between two density matrices $\rho$ and $\sigma$ results equal to the \textit{trace distance}:
\begin{eqnarray*}
	K(\rho||\sigma)=\frac{1}{2}\textrm{Tr}|\rho-\sigma|.
\end{eqnarray*}
Consequently, the corresponding inequality for the Kolmogorov notion of distinguishability is
\begin{eqnarray}\label{eq:resultK}
K(P || p \times q) \leq \sum_i p_i K(\rho_i || \rho)=\mathcal{X}_K.
\end{eqnarray}
It is noteworthy that even when the states in the set do not commute, it is possible to reach the equality in \eqref{eq:resultK}, as will be shown in section \eqref{sec:Example}, by analysing the case of an ensemble of qubits.  This a quite different behavior with respect to the conventional Holevo bound.
\par
\vspace{0.5cm}
\noindent The probability of error between two distributions $p$ and $q$ is
\begin{eqnarray*}
	PE(p||q)=\frac{1}{2}\sum_i \min \{ p_i,q_i \},
\end{eqnarray*}
and it is related to the Kolmogorov notion through:
\begin{eqnarray*}
	PE(p||q)=\frac{1}{2}-\frac{1}{2}K(p||q),
\end{eqnarray*}
therefore, the probability of error between quantum states is
\begin{eqnarray*}
	PE(\rho||\sigma)=\frac{1}{2}-\frac{1}{2}K(\rho||\sigma).
\end{eqnarray*}
Then, we have an alternative interpretation of the inequality \eqref{eq:resultK} within the probability of error notion:
\begin{eqnarray}\label{eq:resultPE}
PE(P || p \times q) \geq \sum_i p_i PE(\rho_i || \rho)=\mathcal{X}_{PE}.
\end{eqnarray}
Finally, it remains to consider Bhattacharyya coefficient between two distributions $p$ and $q$ is defined by
\begin{eqnarray*}
	B(p||q)=\sum_i \sqrt{p_i q_i},
\end{eqnarray*}
wich in extended in the quantum realm as
\begin{eqnarray*}
	B(\rho||\sigma)=\textrm{Tr}\left(\sqrt{\sqrt{\sigma}\rho\sqrt{\sigma}}\right)=\sqrt{F(\rho||\sigma)},
\end{eqnarray*}
being $F(\rho||\sigma)$ the quantum fidelity \cite{ZycowskyLibro,Fusch}. At the same time, this coefficient is related to the squared \textit{Bures/Hellinger distance}:
\begin{eqnarray*}
	d^2_B(\rho||\sigma)=2\left[1-B(\rho||\sigma)\right].
\end{eqnarray*}
The Bures distance satisfies the conditions established in Theorem \ref{th:theorem}, driving to an inequality::
\begin{eqnarray}\label{eq:resultB}
B(P || p \times q) \geq \sum_i p_i B(\rho_i || \rho)=\mathcal{X}_{B}.
\end{eqnarray}
Bearing in mind that $F(\rho||\sigma)\geq\textrm{Tr}(\rho\sigma)$ it is easy to see that
\begin{eqnarray}\label{eq:resultB2}
B(P || p \times q) \geq \gamma[\rho],
\end{eqnarray}
with $\gamma[\rho]=\textrm{Tr}[\rho^2]$ the \textit{purity} of the state which is (a priori) delivered to Bob.
\section{Qubits ensembles} \label{sec:QE}

Now, we shall consider an ensembles of states of qubits, i.e., $\dim{\mathcal{H}_Q}=2$ [cf. Sec. \ref{sec:CommunSch}] and we shall restrict our calculations to von Neumann measurements. Our purpose here is to identify the difference between usual quantities given by the \textit{relative entropy}, i.e. $I_{\mathcal{S}_r}$ and $\mathcal{X}_{\mathcal{S}_r}$, and the generalized quantities DBHQ and GAI. The von Neumann measurements have an advantage mainly because of their Bloch representation.\par
Let $V$ be a unitary operator given by
\begin{eqnarray*}
	V=\vec{s}\cdot (\mathbb{I},i\vec{\sigma}),
\end{eqnarray*}
with $\vec{s}\in \Gamma$, $\Gamma =\{\vec{s}\in\mathbb{R}^4 \ / \ s_0^2+s_1^2+s_2^2+s_3^2=1\}$ and $\vec{\sigma}=(\sigma_1,\sigma_2,\sigma_3)$ the Pauli matrices. The Bob measurements over the system $\mathcal{H}_Q$ are given by
\begin{eqnarray}\label{eq:measurementQubit}
\mathcal{E}=\{E_j\}_{j=0}^1,
\end{eqnarray} 
being $E_j=V^*\ket{j}\bra{j}V$, and $\{\ket{j}\}_{j=0}^1$ the computational basis \cite{Luo08b}.\par
So the  Bloch vector of the events $E_j$ are given by $\textrm{Tr}(E_j \vec{\sigma})=(-1)^j\vec{z}$ being $\vec{z}=(z_1(\vec{s}),z_2(\vec{s}),z_3(\vec{s}))$ where
\begin{eqnarray}
z_1(\vec{s})=2(-s_0s_2+s_1s_3)\label{eq:z1},\\
z_2(\vec{s})=2(s_0s_1+s_2s_3)\label{eq:z2},\\
z_3(\vec{s})=s_0^2+s_3^2-s_1^2-s_2^2,\label{eq:z3}
\end{eqnarray}
Thus,
\begin{eqnarray}
E_{j}=\frac{1}{2}\left[\mathbb{I}+ (-1)^j\vec{z}\cdot\vec{\sigma} \right]. \label{eq:BlochVecE}
\end{eqnarray}
The direction given by $\vec{z}$ characterizes the measurement $\mathcal{E}$. It is easy to see that $\vec{z}$ overspreads the Bloch sphere, i.e., it is possible to measure in any direction.

Now, if the states of the ensemble are given by
\begin{eqnarray}
\rho_{i}=\frac{1}{2}\left(\mathbb{I}+ \vec{\beta}_i\cdot\vec{\sigma} \right), \label{eq:BlochVec} \\
\rho=\frac{1}{2}\left(\mathbb{I}+ \vec{\beta}_m\cdot\vec{\sigma} \right), \label{eq:BlochVecm}
\end{eqnarray}
where $\rho=\sum_i p_i \rho_i$ and therefore $\vec{\beta}_m=\sum_i p_i \vec{\beta}_i$, it can be seen that the probabilities $P_{ij}$ and $p_i q_j$ [cf. Sec. \ref{sec:CommunSch}] are given by:
\begin{eqnarray*}
	P_{ij}=\frac{p_i}{2}\left[1+(-1)^j \vec{\beta}_i \cdot \vec{z}\right],  \\
	p_iq_j=\frac{p_i}{2}\left[1+(-1)^j \vec{\beta}_m \cdot \vec{z}\right].
\end{eqnarray*}

Given an ensemble of operators $\{\rho_1, \rho_2,\dots,\rho_n\}$, with probability $p_i$ each of ones, the non-commutativity $\mathcal{N}_c$ is defined as \cite{Guo}: 
\begin{eqnarray*}
	\mathcal{N}_c&=\frac{1}{2} \sum_{k,l=0}^{n} \begin{Vmatrix}
		[p_k\rho_k , p_l \rho_l ]
	\end{Vmatrix}_2,
\end{eqnarray*}
where $||A||_2=\sqrt{\Tr{A^\dagger A}}$ is the Hilbert-Schmidt norm. After some algebra, we have:
\begin{eqnarray} \label{eq:NC}
\mathcal{N}_c= \sum_{k,l=0}^{n}\frac{p_k p_l}{2\sqrt{2}}  \sqrt{\left|\vec{\beta}_k\right|\left|\vec{\beta}_l\right|^2 - \left(\vec{\beta}_k\cdot\vec{\beta}_l\right)^2}= \sum_{k,l=0}^{n} \frac{p_k p_l}{2\sqrt{2}}\left| \vec{\beta}_k \times \vec{\beta}_l \right|.
\end{eqnarray}
In the case of the state \eqref{eq:BlochVecm}, the purity $\gamma$ results:
\begin{eqnarray} \label{eq:purity}
\gamma[\rho]=\textrm{Tr}\! \left[\rho^2\right] = \frac{1}{2} \left(1+|\vec{\beta}_m|^2\right).
\end{eqnarray}
Now, we are in position of obtaining the generalized expressions for the distance based Holevo quantity $\mathcal{X}_d$ and $D(P||p\times q)$ for the cryptographic distinguishability measures considered in Sec. \ref{sec:QuantDist2}. We have omitted the probability of error (PE) because the expressions are totally analogous to the Kolmogorov notion (K).

\vspace{0.5cm}

For the \textit{Kolmogorov distance}
\begin{eqnarray}
\mathcal{X}_K=\frac{1}{2}\sum_i p_i \left|\vec{\beta}_i-\vec{\beta}_m\right|, \label{eq:Xk} \\
K(P||p\times q)=\frac{1}{2}\sum_i \left|(\vec{\beta}_i-\vec{\beta}_m)\cdot \vec{z}\right|,
\end{eqnarray}
and for the \textit{Bhattacharyya coefficient}
\begin{eqnarray}
&\mathcal{X}_B= \frac{1}{\sqrt{2}} \sum_i p_i  \sqrt{1+\vec{\beta}_i\cdot\vec{\beta}_m + \sqrt{\left(1-\left|\vec{\beta}_i\right|^2\right)\left(1-\left|\vec{\beta}_m\right|^2\right)}}, \label{eq:XB}\\
&B(P||p\times q)= \nonumber\\
&= \sum_i \frac{p_i}{\sqrt{2}}  \sqrt{1+(\vec{\beta}_i\cdot\vec{z})(\vec{\beta}_m\cdot\vec{z})  + \sqrt{\left[1-\left(\vec{\beta}_i\cdot \vec{z}\right)^2\right]\left[1-\left(\vec{\beta}_i\cdot \vec{z}\right)^2\right]}}. \label{eq:IB}
\end{eqnarray}

In the case of the \textit{Relative entropy}, the Holevo bound and $H_r(P||p\times q)$ take the form
\begin{eqnarray}
\mathcal{X}_{S_r}=\frac{1}{2} \left[ \sum_i p_i f\!\left(\left|\vec{\beta}_i\right|\right)\right] -\frac{1}{2} f\!\left(\left|\vec{\beta}_m\right|\right) \label{eq:Xs}, \\
H_r(P||p\times q)=\frac{1}{2} \left[ \sum_i p_i f\!\left(\left|\vec{\beta}_i\cdot\vec{z}\right|\right)\right] -\frac{1}{2} f\!\left(\left|\vec{\beta}_m\cdot\vec{z}\right|\right), \label{eq:Ih}
\end{eqnarray}
where $f(x)=(1+x)\log_2(1+x)+(1-x)\log_2(1-x)$ is minus the binary Shannon entropy.
\vspace{0.5cm}

Thus, once obtained the analytical expressions for the Kolmogorov notion, we can establish the following result:
\begin{theorem}
	For any two qubit ensemble (see \eqref{eq:BlochVec} and \eqref{eq:BlochVecm}, with $n=2$) it holds:
	\begin{eqnarray*}
		I_K(X,Y)=\max_{\mathcal{E}} K(P,p\times q)=\mathcal{X}_K.
	\end{eqnarray*}
\end{theorem}
\begin{proof} By definition, we have
	\begin{eqnarray*}
		\vec{\beta}_m=\hat{p}\vec{\beta}_0+(1-\hat{p})\vec{\beta}_1, \\
		K(P,p\times q)=\frac{1}{2}\hat{p}\left|(\vec{\beta}_0-\vec{\beta}_m)\cdot \vec{z}\right|+\frac{1}{2}(1-\hat{p})\left|(\vec{\beta}_1-\vec{\beta}_m)\cdot \vec{z}\right|, \\
		\mathcal{X}_K=\frac{1}{2}\hat{p}\left|\vec{\beta}_0-\vec{\beta}_m\right|+\frac{1}{2}(1-\hat{p})\left|\vec{\beta}_1-\vec{\beta}_m\right|.
	\end{eqnarray*}
	Choosing $\vec{z}=\frac{\vec{\beta}}{|\vec{\beta}|}$   where $\vec{\beta}=\vec{\beta}_0-\vec{\beta}_1$ it follows
	\begin{eqnarray*}
		\mathcal{X}_K=\hat{p}(1-\hat{p})\vec{\beta}=K(P,p\times q)\rvert_{\vec{z}=\frac{\vec{\beta}}{|\vec{\beta}|}}.
	\end{eqnarray*}
	Therefore, $I_K=K(P,p\times q)\rvert_{\vec{z}=\frac{\vec{\beta}}{|\vec{\beta}|}}$.
\end{proof}
Finally, we will study the behavior of $I_d$ and $\mathcal{X}_d$ (for Kolmogorov and Bhattacharyya notions) in contrast with the known quantities $I_{\mathcal{S}_r}$ and $\mathcal{X}_{\mathcal{S}_r}$, using the non-commutativity measure $\mathcal{N}_c$ and the purity $\gamma$ as figure of merit for the ensemble properties.
\subsection{Example} \label{sec:Example}
\begin{figure}
	\centering 
	\includegraphics[width=0.8\columnwidth]{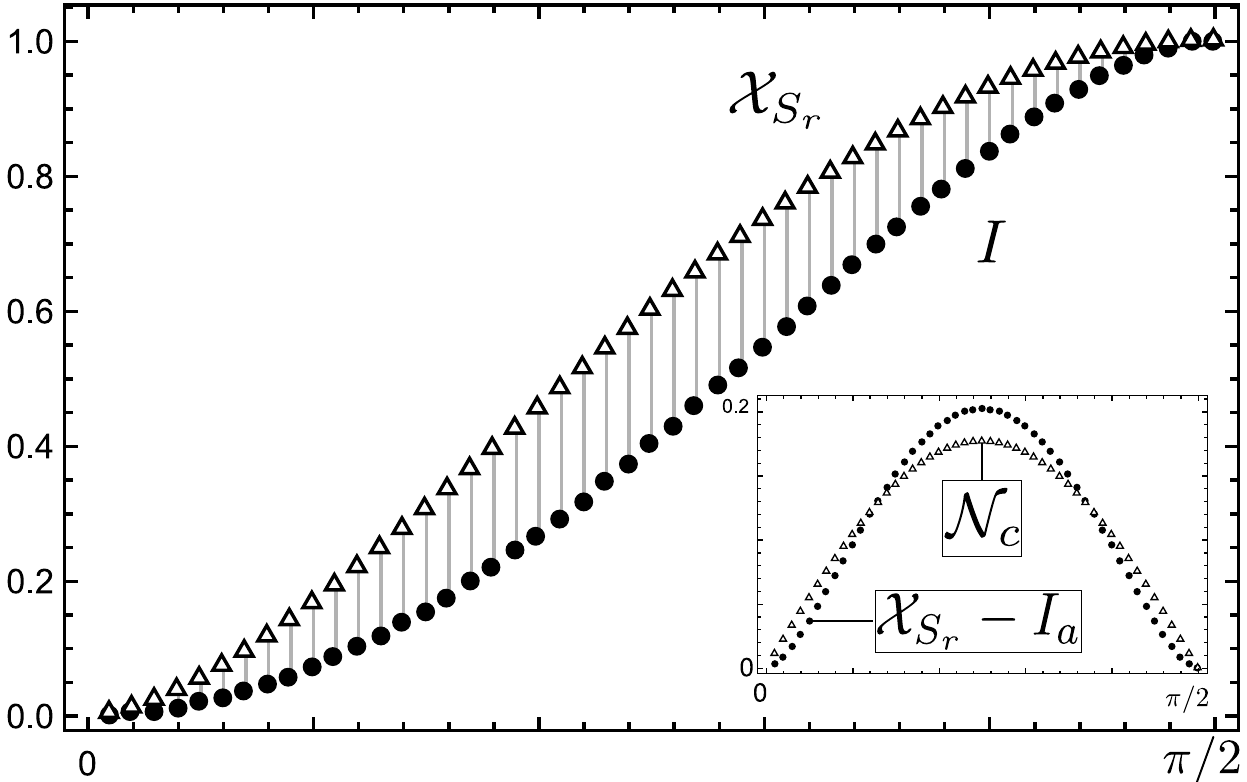}
	\caption[Fig1]{Comparative of $I$, $\mathcal{X}_{\mathcal{S}_r}$ and $\mathcal{N}_c$ for the ensemble \eqref{eq:ensemble0} and $\eqref{eq:ensemble1}$ and $\hat{p}=1/2$. All the quantities shown are dimensionless.} \label{fig:K}
\end{figure}
Consider an ensemble composed by two pure states $\{\rho_0,\rho_1\}$, being \cite{Nielsen&Chuang,ZeroError,ZycowskyLibro}
\begin{eqnarray}\label{eq:ensemble0}
\rho_0=
\begin{bmatrix}
1 & 0 \\
0 & 0 
\end{bmatrix},
\end{eqnarray}
\begin{eqnarray}\label{eq:ensemble1}
\rho_1=
\begin{bmatrix}
\cos^2\theta & \cos\theta\sin\theta \\
\cos\theta\sin\theta & \sin^2\theta 
\end{bmatrix},
\end{eqnarray}
with probabilities $p_0=\hat{p}$ and $p_1=1-\hat{p}$. For $\theta=0$, $\rho_0=\rho_1$ and for $\theta=\frac{\pi}{2}$, the states commute. The corresponding Bloch vectors are [cf. \eqref{eq:BlochVec} and \eqref{eq:BlochVecm}]
\begin{eqnarray}
\vec{\beta}_0&=(0,0,1), \label{eq:beta0}\\
\vec{\beta}_1&=\left(\sin 2\theta,0,\cos 2\theta\right), \label{eq:betai}\\
\vec{\beta}_m&=\left( (1-\hat{p})\sin 2\theta ,0, \hat{p} + (1-\hat{p})\cos 2\theta\right). \label{eq:betam}
\end{eqnarray}

\begin{figure}
	\centering 
	\includegraphics[width=0.8\columnwidth]{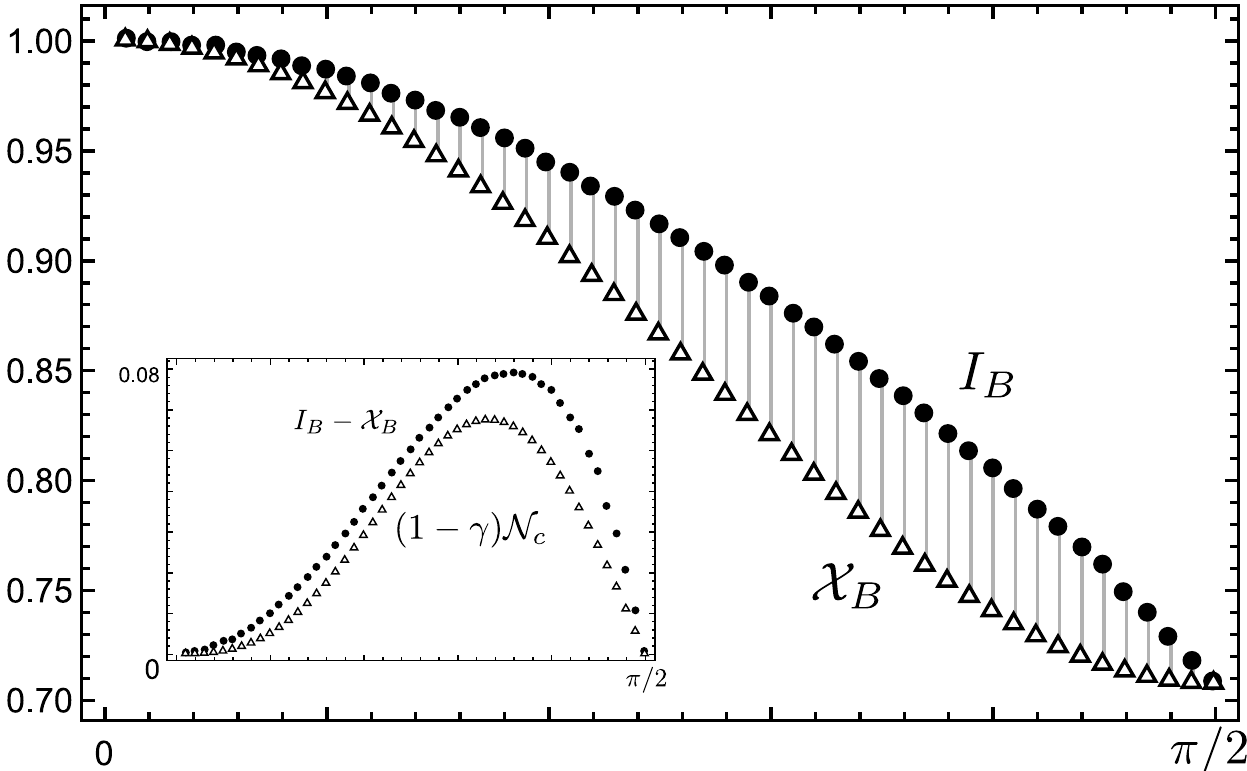}
	\caption[Fig2]{Comparative of $I_B$, $\mathcal{X}_{B}$, $\gamma$ and $\mathcal{N}_c$ for the ensemble \eqref{eq:ensemble0} and $\eqref{eq:ensemble1}$ and $\hat{p}=1/2$. All the quantities shown are dimensionless.} \label{fig:B}
\end{figure}
Consequently, the non-commutativity measure $\mathcal{N}_c$ and the purity $\gamma$ take the following analytical forms [cf. \eqref{eq:NC} and \eqref{eq:purity}]:
\begin{eqnarray*}
	\mathcal{N}_c = \frac{\hat{p}(1-\hat{p})}{\sqrt{2}}\left|\sin 2 \theta\right|,
\end{eqnarray*}

\begin{eqnarray*}
	\gamma[\rho] = 1-2\hat{p}(1-\hat{p}) (\sin \theta )^2.
\end{eqnarray*}
Additionally, the expression for $\mathcal{X}_K$ (also for $I_K$) is
\begin{eqnarray*}
	\mathcal{X}_K=2\hat{p}(1-\hat{p})|\sin \theta|.
\end{eqnarray*}

The remaining cases, i.e. $\mathcal{X}_d$ and $I_d$ for relative entropy $\mathcal{S}_r$ and Bhattacharyya coefficient $B$, constitute a more interesting case. We evaluate the generalized accessible information for each cases [cf. \eqref{eq:IB} and \eqref{eq:Ih}] (with $\hat{p}=\frac{1}{2}$) for $\theta\in[0,\frac{\pi}{2}]$. On the other hand, inserting \eqref{eq:beta0}-\eqref{eq:betam} in \eqref{eq:XB} and \eqref{eq:Xs} we can obtain the analytic expressions for the distance based Holevo quantity. 

Figure \ref{fig:K} shows the behavior of $\mathcal{X}_{\mathcal{S}_r}$ and $I_{\mathcal{S}_r}$. The difference between these quantities $\mathcal{X}_{\mathcal{S}_r}-I_{\mathcal{S}_r}$ increases and decreases as the non-commutativity measure $\mathcal{N}_c$ does.

Figure \ref{fig:B} displays $\mathcal{X}_{B}$ and $I_B$. In this case, the difference $I_B-\mathcal{X}_{B}$ does not have the same behavior than $\mathcal{N}_c$. By plotting $(1-\gamma)\mathcal{N}_c$ it is straightforward to conclude that $I_B-\mathcal{X}_{B}$ also contains information about the purity $\gamma$ of the ensemble.

\section{Concluding remarks\label{sec:conclusions}}

In this work, we have proposed a generalization of the Holevo theorem by means of distance measures, focusing our study in cryptographic notions of distinguishability. In particular, we have obtained the corresponding new inequalities to the notions of Kolmogorov, Bhattacharyya coefficient and probability of error. To explore the behaviors of the generalized quantities (DBHQ and GAI) in these three cases, we have calculated the corresponding analytical expressions for the qubit case, using the von Neumann measurements, proving that DBHQ and GAI are equal for the Kolmogorov notion of distinguishability ($I_K=\mathcal{X}_K$) for any qubit ensemble of two elements. We have also obtained the analytical expressions for the non-commutativity (measured by $\mathcal{N}_c$) and purity (measured by $\gamma$) of the ensemble as a function of the Bloch vectors of the states composing the ensemble. Finally, we have considered an ensemble of two pure states in which we have numerically computed the accessible information $I_a$ and GAI for the Bhattacharyya notion of distinguishability. By using the non-commutativity  and the purity as figure of merit of the ensemble properties, we found that that:  1) $\mathcal{X}_{\mathcal{S}_r}-I_{\mathcal{S}_r}$ increase and decrease if the non-commutativity does and 2) the difference of the generalized quantities for the Bhattacharyya notion, i.e. $I_B-\mathcal{X}_B$, captures not only the non-commutativity of the ensemble but also the purity showing a richer behavior that the relative entropy case.

It would be further interesting to deepen the study of GAI and DBHQ meaning, the tightness of the bounds and his operational uses, for different distance measures and different distinguishability notions. 

\begin{acknowledgements}
The authors acknowledge Secyt-UNC-Argentina for financial support. D. G. Bussandri is a fellowship holder from CONICET.
\end{acknowledgements}

{}
\end{document}